\documentclass[letterpaper, 10 pt, conference]{ieeeconf}  

\IEEEoverridecommandlockouts                              
\usepackage[cmex10]{amsmath}  
\usepackage{amsfonts}
\usepackage{graphicx}
\usepackage{psfrag}
\usepackage{threeparttable}  
\usepackage{booktabs}  

\newcommand{\mat}[1]{\boldsymbol{#1}}
\newcommand{\vect}[1]{\boldsymbol{#1}}
\newcommand{\snr}{\mathrm{SNR}}
\newcommand{\integers}{{\mathbb{Z}}}
\newcommand{\complex}{{\mathbb{C}}}
\newcommand{\ie}{{\it i.e.}}

\newcommand{\Hcal}{\mathcal{H}}
\newcommand{\ones}{\mathbf 1}

\newtheorem{theorem}{Theorem}
\newtheorem{lemma}[theorem]{Lemma}
\newtheorem{corollary}[theorem]{Corollary}
\newtheorem{definition}{Definition}
\newtheorem{example}{Example}

\begin{document}

\title{\LARGE \bf
Delay-Rate Tradeoff for Ergodic Interference Alignment in the Gaussian Case
}


\author{Joseph~C.~Koo, William Wu, and John~T.~Gill,~III
\thanks{J.~Koo and J.~Gill are with the Department of Electrical Engineering,
	Stanford University, Stanford, CA  94305, USA.  W.~Wu was with
	the Department of Electrical Engineering, Stanford University.
	He is now affiliated with the Jet Propulsion Laboratory,
	California Institute of Technology, Mail Stop 238-420, 4800
	Oak Grove Drive, Pasadena, CA  91109-8099.
        E-mails: {\tt\small \{jckoo, gill\}@stanford.edu, \tt\small
        William.Wu@jpl.nasa.gov}}%
}

\maketitle
\thispagestyle{empty}
\pagestyle{empty}

\begin{abstract}
In interference alignment, users sharing a wireless channel are each
able to achieve data rates of up to half of the non-interfering
channel capacity, no matter the number of users.  In an ergodic
setting, this is achieved by pairing complementary channel
realizations in order to amplify signals and cancel interference.
However, this scheme has the possibility for large
delays in decoding message symbols.  We show that delay can be
mitigated by using outputs from potentially \emph{more than two}
channel realizations, although data rate may be reduced.  We further
demonstrate the tradeoff between rate and delay via a time-sharing
strategy.  Our analysis considers Gaussian channels; an extension to
finite field channels is also possible.
\end{abstract}

\section{Introduction}
\label{sec:intro}

The technique of interference alignment has expanded what is known
about achievable rates for wireless interference channels.  First
proposed by Maddah-Ali \textit{et
al.}~\cite{maddah-ali:mimo_x_channels} and then applied to wireless
interference channels by Cadambe and
Jafar~\cite{cadambe:int_align_K_user_int_channel}, interference
alignment employs a transmission strategy that compensates for the
interference channel between transmitters and receivers.  At each
receiver, the interference components can then be consolidated into a
part of the channel that is orthogonal to the signal component.  In
fact, the interference is isolated to half of the received signal
space, while the desired signal is located in the other half---leading
to the statement that every receiver can have ``half the cake.''  This
is a significant improvement over every receiver receiving only $1/K$
of the cake, which is the case if standard orthogonalization
techniques are used (where $K$ is the number of transmitter-receiver
pairs).

Interference alignment in an ergodic setting is studied in Nazer
\textit{et al.}~\cite{nazer:ergodic_int_align}, and provides the basis
for our analysis.  Using their Gaussian achievable scheme, we delve
deeper into the associated decoding delays and consider how delays may
be reduced, although at the cost of decreased rate.  Even though the
analysis in \cite{nazer:ergodic_int_align} additionally considers a
scheme for finite field channels (also similar to the method in
\cite{jeon:capacity_class_multisource_relay_networks}), we defer to
the reader the extension of our analysis to the finite field case.

Our approach for reducing delays is to consider interference alignment
where alignment may require more than one additional instance of
channel fading.  In \cite{nazer:ergodic_int_align}, interference is
aligned by transmitting the same message symbol during complementary
channel realizations.  In contrast, our approach will utilize multiple
channel realizations (potentially more than two), which when summed
together yield cancelled interference (and amplified signal).  We call
such a set of channel matrices an \emph{alignment set}---which will be
more formally defined later.  Using multiple channel realizations to
align interference has also been studied in
\cite{nazer:int_align_general_message_sets} for different cases of
receiver message requirements; however, we instead consider how to
utilize these many channel realizations to reduce the delay of
individual messages at each receiver.  At first glance, it may seem
that using alignment sets of larger sizes will only increase the
delay; but if we allow alignment using alignment sets of multiple
sizes simultaneously, then we can decrease the time required for a
message symbol to be decoded.

We now give a simple example of an alignment set and show the concept
of ergodic interference alignment.
\begin{example}
\label{ex:align_set_4}
Consider a $3$-user Gaussian interference channel with channel
response given by $\vect{Y} = \mat{H} \vect{X} + \vect{Z}$, where
$\vect{X}$ denotes the transmitted symbols (with power constraint
$E[|X_k|^2] \leq P$ for each user $k = 1,2,3$), $\mat{H}$ is the
channel matrix, $\vect{Z}$ is independently and identically
distributed zero-mean unit-variance additive white Gaussian noise, and
$\vect{Y}$ gives the received symbols.  Suppose the following channel
matrices occur at time steps $t_0$, $t_1$, $t_2$, and $t_3$,
respectively:
\[
\small
\begin{array}{ll}
\mat{H}^{(0)} = \left[ \begin{array}{rrr}
1 & -1 & 1 \\ 1 & 1 & -1 \\ -1 & 1 & 1
\end{array} \right] &
\mat{H}^{(1)} = \left[ \begin{array}{rrr}
1 & -1 & -1 \\ -1 & 1 & 1 \\ 1 & 1 & 1
\end{array} \right] \\ \\
\mat{H}^{(2)} = \left[ \begin{array}{rrr}
1 & 1 & -1 \\ -1 & 1 & 1 \\ -1 & -1 & 1
\end{array} \right] &
\mat{H}^{(3)} = \left[ \begin{array}{rrr}
1 & 1 & 1 \\ 1 & 1 & -1 \\ 1 & -1 & 1
\end{array} \right]
\end{array}
\mbox{.}
\]
If the same [complex] vector~$\vect{X}$ is sent at all these times,
then the sum of the non-noise terms is given by $\sum_{i=0}^3
\mat{H}^{(i)} \vect{X} = 4 [X_1, X_2, X_3]^T$ because $\sum_{i=0}^3
\mat{H}^{(i)} = 4 \mat{I}$.  By utilizing all four channel
realizations together, the signals (diagonal entries) are amplified,
while the interference terms (off-diagonal entries) are cancelled, so
this collection of matrices is an alignment set.  As long as a
receiver knows when an alignment set occurs, then in order to decode
his own message, he does not need to know the channel fades to the
other receivers.
\end{example}

Inferring from \cite{cadambe:multiple_access_outerbounds} or
\cite{cadambe:parallel_int_channels_not_always_separable}, the astute
reader may notice that in the example, the sum capacity when sending
across each channel matrix separately is actually greater than the
alignment rate---a capacity of $4 \log (1 + 3P)$ for separate coding,
compared to a rate of $3 \log (1 + 4 P)$ by using the indicated
interference alignment scheme.  However, when the number of
transmitters (and receivers) exceeds the number of alignment channel
realizations, then the rate benefits of using alignment sets start to
become evident.  Aligning across $4$ channel realizations with $K$
transmitter-receiver pairs, a rate of $K \log (1 + 4P)$ is achievable,
which can quickly eclipse the separate-coding sum capacity of $4 \log
(1 + K P)$.  Moreover, as we will discuss, the benefit of using larger
alignment sets is not in the rate, but rather in the reduction of
decoding delay.

In the next section, we will formally describe the interference
alignment setup, and define our notions of rate and delay.  In
Section~\ref{sec:align_complement}, we will take a brief look at the
conventional ergodic interference alignment scheme, by considering the
rate and delay inherent in aligning interference using complementary
channel realizations.  Section~\ref{sec:align_multiple} will give the
main result of this work, which is the analysis of rate and delay when
aligning interference by utilizing multiple channel realizations.  We
will also give a scheme for trading off the rate and the delay.  We
conclude in Section~\ref{sec:conclusion}.

\section{Preliminaries}
\label{sec:prelim}

The setup is the same as the $K$-user interference channel of
\cite{nazer:ergodic_int_align} and
\cite{nazer:int_align_general_message_sets}, where there are $K$
transmitter-receiver pairs.  The number of channel uses is $n$.  For
the $k$-th transmitter, $k = 1,\ldots,K$, each message~$w_k$ is chosen
independently and uniformly from the set $\{1, 2, \ldots, 2^{n
\tilde{R}_k}\}$ for some $\tilde{R}_k \geq 0$.  Only transmitter~$k$
knows message~$w_k$.  Let $\mathcal{X}$ be the channel input and
output alphabet.  The message~$w_k$ is encoded into the $n$ channel
uses using the encoder $\mathcal{E}_k : \{1, 2, \ldots, 2^{n
\tilde{R}_k}\} \to \mathcal{X}^n$.  The output of the encoding
function is the transmitted symbol $X_k(t) = [\mathcal{E}_k(w_k)]_t$
at time~$t$, for $t = 1,\ldots,n$.

The communication channel undergoes fast fading, so the channel fades
change at every time step.  At time~$t$, the channel
matrix~$\mat{H}(t)$ has complex entries $[\mat{H}(t)]_{kl} =
h_{kl}(t)$ for $k,l = 1,\ldots,K$.  In this model, all transmitters
and receivers are given perfect knowledge of $\mat{H}(t)$ for all
times~$t$.  We call $\Hcal$ to be the set of all possible channel
fading matrices.

The message symbol~$X_k(t)$ is transmitted at time~$t$.  We assume
zero delay across the channel, so the channel output seen by
receiver~$k$ at time~$t$ is the received symbol
\begin{equation}
Y_k(t) = \sum_{l=1}^K h_{kl}(t) X_l(t) + Z_k(t)
\mbox{,}
\label{eq:channel_model}
\end{equation}
where $Z_k(t)$ is an additive noise term.  Each receiver~$k$ then
decodes the received message symbols according to $\mathcal{D}_k :
\mathcal{X}^n \to \{1, 2, \ldots, 2^{n \tilde{R}_k}\}$, to produce an
estimate~$\hat{w}_k$ of $w_k$.
\begin{definition}
The ergodic rate tuple $(R_1, R_2, \ldots, R_K)$ is \emph{achievable}
if for all $\epsilon > 0$ and $n$ large enough, there exist channel
encoding and decoding functions $\mathcal{E}_1, \ldots,
\mathcal{E}_K$, $\mathcal{D}_1, \ldots, \mathcal{D}_K$ such that
$\tilde{R}_k > R_k - \epsilon$ for all $k = 1,2,\ldots,K$, and $P
\left( \bigcup_{k=1}^K \{\hat{w}_k \ne w_k\} \right) < \epsilon$.
\end{definition}

We assume a Gaussian channel with complex channel inputs and outputs,
so $\mathcal{X} = \complex$.  Each transmitter~$k$ has power
constraint
\[
E[ |X_k(t)|^2] \leq \snr_k
\mbox{,}
\]
where $\snr_k \geq 0$ is the signal-to-noise ratio.  The channel
coefficients~$h_{kl}(t)$, $k,l = 1,\ldots,K$, are independently and
identically distributed both in space and time.  We require also that
$h_{kl}$ be drawn from a distribution which is symmetric about zero,
so $P(h_{kl}) = P(-h_{kl})$.  The noise terms~$Z_k(t)$ are drawn
independently and identically from a circularly-symmetric complex
Gaussian distribution; thus, $Z_k(t) \sim \mathcal{CN}(0,1)$.

\subsection{Channel Quantization}
\label{subsec:channel_quant}

In this exposition, we consider quantized versions of the channel
matrix.  For some quantization parameter $\gamma > 0$, let
$Q_\gamma(h_{kl})$ be the closest point in $(\integers + j \integers)
\gamma$ to $h_{kl}$ in Euclidean distance.  The $\gamma$-quantized
version of the channel matrix $\mat{H} \in \complex^{K \times K}$ is
given by the entries $[\mat{H}_\gamma]_{kl} = Q_\gamma(h_{kl})$.

Our scheme uses typical realizations of the channel matrices.  For any
$\epsilon > 0$, choose the maximum magnitude $\tau > 0$ such that
$P(\bigcup_{k,l} \{|h_{kl}| > \tau \}) < \frac{\epsilon}{3}$.  Throw
out all time indices with any channel coefficient magnitude larger
than $\tau$. Let $\gamma$ and $\delta$ be small positive constants.
Then choose $n$ large enough so that the typical set of sequences
$A_\delta^n$ of channel matrices has probability $P(A_\delta^n) \geq
1-\frac{\epsilon}{3}$ (see \cite{nazer:ergodic_int_align} for
details).  Because this sequence of $\gamma$-quantized channel
matrices is $\delta$-typical, the corresponding rate decrease is no
more than a fraction of $\delta$.

In the remainder of this paper, we will only deal with the
$\gamma$-quantized channel matrices~$\mat{H}_\gamma$, so we drop the
subscript~$\gamma$; all further occurrences of $\mat{H}$ refer to the
quantized channel realization~$\mat{H}_\gamma$.  We also redefine the
channel alphabet~$\Hcal$ to only include the typical set of quantized
channel matrices, which has cardinality $|\Hcal| = (2 \tau /
\gamma)^{2K^2}$.

\subsection{Aligning Interference}
\label{subsec:align_int}

In the standard interference alignment approach, the interference is
aligned by considering the channel matrix~$\mat{H}$ in tandem with its
complementary matrix~$\mat{H}^c$, where
\[
\mat{H}^c = \left[ \begin{array}{r@{}lr@{}lcr@{}l}
  & h_{11}    & - & h_{12}    & \cdots & - & h_{1K} \\
- & h_{21}    &   & h_{22}    & \cdots & - & h_{2K} \\
  & \; \vdots &   & \; \vdots & \ddots &   & \; \vdots \\
- & h_{K1}    & - & h_{K2}    & \cdots &   & h_{KK} \\
\end{array} \right]
\mbox{.}
\]
That is, $\mat{H}^c$ has entries $h_{kl}$ for $k = l$ and $-h_{kl}$
for $k \ne l$.

For alignment using more channel realizations, we define the concept
of an alignment set.
\begin{definition}
\label{def:align_set}
An \emph{alignment set} of size $m \in 2 \integers^+$ is a collection
of matrices $\mathcal{A} = \{\mat{H}^{(0)}, \mat{H}^{(1)}, \ldots,
\mat{H}^{(m-1)}\}$ such that the diagonal entries (signal terms) are
the same:
\begin{equation}
h^{(0)}_{kk} = h^{(1)}_{kk} = \cdots = h^{(m-1)}_{kk}
\label{eq:align_set_def_signal}
\end{equation}
for $k = 1,\ldots,K$, and the sum of interference terms cancel:
\begin{equation}
\left| h^{(0)}_{kl} \right| = \left| h^{(1)}_{kl} \right| = \cdots =
\left| h^{(m-1)}_{kl} \right|
\label{eq:align_set_def_interference_magnitude}
\end{equation}
and
\begin{eqnarray}
\left| \{h^{(i)}_{kl} = h^{(0)}_{kl} \; | \; i=1,\ldots,m-1\} \right|
& = & \frac{m}{2} - 1  \label{eq:align_set_def_interference_pos} \\
\left| \{h^{(i)}_{kl} = -h^{(0)}_{kl} \; | \; i=1,\ldots,m-1\} \right|
& = & \frac{m}{2}  \label{eq:align_set_def_interference_neg}
\end{eqnarray}
for $k = 1,\ldots,K$, $l = 1,\ldots,K$, $k \ne l$.  Within an
alignment set, the sum of channel matrices, $\mat{B} =
\sum_{i=0}^{m-1} \mat{H}^{(i)}$, will have entries $b_{kk} = m
h_{kk}^{(0)}$ and $b_{kl} = 0$, for $k,l = 1,\ldots,K$, $k \ne l$.
We denote $\mathcal{A}_{\mat{H}}$ to be an alignment set of which
$\mat{H}$ is a member.
\end{definition}
 
We have seen some examples of alignment sets already.  Any channel
realization~$\mat{H}$ and its complement~$\mat{H}^c$ together form an
alignment set of size~$2$.  Additionally, the set of matrices given in
Example~\ref{ex:align_set_4} is an alignment set of size~$4$.

Since channel transmission is instantaneous, the only delay considered
is due to waiting for the appropriate channel realizations before a
message symbol can be decoded.
\begin{definition}
The \emph{average delay} of an ergodic interference alignment scheme
is the expected number of time steps between the first instance a
message symbol~$\vect{X}$ is sent and the time until $\vect{X}$ is
recovered at the receiver.
\end{definition}

If $\vect{X}(t_0)$ is sent at time~$t_0$ but can not be decoded until
the appropriate interference alignment occurs at time~$t_1$, then the
delay is $t_1 - t_0$.  Note that the delay does not consider the
decoding of the entire message~$w_k$---just the symbols transmitted at
each individual time, $X_k(t)$, $k = 1,\ldots,K$.

\section{Interference Alignment using Complementary Channel Realization}
\label{sec:align_complement}

The method of interference alignment via sending the same channel
input vector when a complementary channel realization occurs is given
in \cite{nazer:ergodic_int_align}.  Call $R_k^{(2)}$ the achievable
rate for interference alignment using complements (\ie, requiring two
channel realizations before decoding each message symbol).

\begin{lemma}[{\cite[Theorem~3]{nazer:ergodic_int_align}}]
\label{thm:rate_align_2}
An achievable rate tuple by aligning using complementary channel
realizations is
\[
\textstyle
R_k^{(2)} = \frac{1}{2} E[\log (1 + 2 |h_{kk}|^2 \snr_k)]
\] 
for $k = 1,\ldots,K$, where the expectation is over the distribution
of channel fades~$h_{kk}$ drawn from the matrices in $\Hcal$.
\end{lemma}

When a channel realization~$\mat{H}$ occurs, then the sent message
symbol is decoded when the complementary channel
realization~$\mat{H}^c$ occurs.  Let $d^{(2)}$ denote the average
delay  between channel realizations $\mat{H}$ and $\mat{H}^c$.
\begin{lemma}
\label{thm:delay_align_2}
When all channel realizations are equally likely, the average delay
incurred by interference alignment with complementary channel
realizations is $d^{(2)} = |\Hcal|$.
\end{lemma}
\begin{proof}
Each channel realization is equally likely at each time. The time
until $\mat{H}^c$ occurs is a geometric random variable with parameter
$P(\mat{H}^c) = 1/|\Hcal|$.  The average delay is $|\Hcal|$.
\end{proof}

Note that the delay~$d^{(2)}$ can be quite large.  Using our
quantization scheme, $d^{(2)} = |\Hcal| = (2 \tau / \gamma)^{2K^2}$.

\section{Interference Alignment using Multiple Channel Realizations}
\label{sec:align_multiple}

This section will  focus on using alignment sets of sizes $m = 2$ and
$m = 4$.  Extensions for larger alignment sets will be discussed in
Section~\ref{subsec:align_multiple_larger_align_sets}.

For ease of analysis, we assume that each channel
realization~$\mat{H}$ is equally likely, although the ideas presented
may be readily extended to the cases where the distribution of channel
realizations is non-uniform.  However, for this particular
interference alignment scheme to work, all channel realizations within
the same alignment set must be equiprobable: for an alignment set
$\mathcal{A}_{\mat{H}} = \{\mat{H}, \mat{H}^{(1)}, \mat{H}^{(2)},
\ldots, \mat{H}^{(m-1)}\}$, we require that $P(\mat{H}) =
P(\mat{H}^{(1)}) = P(\mat{H}^{(2)}) = P(\mat{H}^{(m-1)})$.
Fortunately, this holds since we assume that channel entries are drawn
from distributions that are symmetric about zero.

\subsection{First-to-Complete Alignment}
\label{subsec:align_multiple_first_to_complete}

We call the following scheme for achieving lower delay the
\emph{first-to-complete} scheme, which is essentially a
coupon-collecting race between an alignment set of size~$2$ and an
alignment set of size~$4$.  For some channel realization $\mat{H} \in
\Hcal$ (occurring at a time~$t_0$)---since the entire future of
channel realizations is known---we can collect the realizations
occurring at future times $t > t_0$.  Now we say that an alignment
set~$\mathcal{A}_{\mat{H}}$ of size~$4$ has been \emph{completed} once
all matrices $\tilde{\mat{H}} \in \mathcal{A}_{\mat{H}}$ have been
realized.  If $\mat{H}^c$ occurs before $\mathcal{A}_{\mat{H}}$ is
completed, then pair up $\mat{H}$ with that realization of
$\mat{H}^c$.  Otherwise, group together $\mat{H}$ with the other
members of the alignment set $\mathcal{A}_{\mat{H}}$.

We derive the achievable rate by separately finding the rates when
decoding using alignment sets of different sizes, and then weighting
these rates by the probabilities that a particular-sized set is
completed before the other.  From \cite{nazer:ergodic_int_align}, if
$\mat{H}$ at time~$t_0$ is paired with $\mat{H}^c$ at time~$t_1$, then
the same symbol vector~$\vect{X}(t_0)$ is transmitted at both times
$t_0$ and $t_1$.  Since this is alignment with channel complements,
the rate $R_k = \frac{1}{2} E[\log (1 + 2|h_{kk}|^2 \snr_k)] -
\epsilon$ is achievable with probability $1 - \epsilon$.

Now we find the rate when $\mat{H}$ at time $\hat{t}_0 = t_0$ is
instead grouped with the members of its size-$4$ alignment
set~$\mathcal{A}_{\mat{H}}$.  Assume that the channel realizations of
the other members of the alignment set occur at times $\hat{t}_1$,
$\hat{t}_2$, and $\hat{t}_3$, respectively.  In the scheme, we send
the same message symbol~$X_k(\hat{t}_0)$ at times $\hat{t}_0$,
$\hat{t}_1$, $\hat{t}_2$, and $\hat{t}_3$.  The channel outputs are
\begin{equation}
Y_k(t) = h_{kk}(t) X_k(\hat{t}_0) + \sum_{l \ne k} h_{kl}(t)
X_l(\hat{t}_0) + Z_k(t)
\label{eq:channel_output_align_4}
\end{equation}
for $t = \hat{t}_0, \hat{t}_1, \hat{t}_2, \hat{t}_3$.  From the
alignment set definition, we know $h_{kk}(\hat{t}_0) =
h_{kk}(\hat{t}_1) = h_{kk}(\hat{t}_2) = h_{kk}(\hat{t}_3)$ and
$h_{kl}(\hat{t}_0) + h_{kl}(\hat{t}_1) + h_{kl}(\hat{t}_2) +
h_{kl}(\hat{t}_3) = 0$ for $k = 1,\ldots,K$ and $l \ne k$.  Thus, the
signal-to-interference-plus-noise ratio of the channel from
$X_k(\hat{t}_0)$ to $Y_k(\hat{t}_0) + Y_k(\hat{t}_1) + Y_k(\hat{t}_2)
+ Y_k(\hat{t}_3)$ is at least
\[
\frac{\snr_k ((4 |\Re(h_{kk})| - 2 \gamma)^2 + (4 |\Im(h_{kk})| - 2
\gamma)^2)}{4 + (2 \gamma)^2 \sum_{l \ne k} \snr_l}
\mbox{.}
\]
Taking the channel quantization parameter $\gamma \to 0$, the SINR is
$4 |h_{kk}|^2 \snr_k$, which gives the rate (as $\tau \to \infty$):
\begin{equation}
\textstyle R_k = \frac{1}{4} E[\log (1 + 4|h_{kk}|^2 \snr_k)] -
\frac{2 \epsilon}{3}
\mbox{.}
\label{eq:rate_limit_align_4}
\end{equation}
Thus there exist $\gamma$ and $\tau$ such that we achieve $R_k >
\frac{1}{4} E[\log (1 + 4|h_{kk}|^2 \snr_k)] - \epsilon$ with
probability $1 - \epsilon$ when aligning using an alignment set of
size $4$.\footnotemark
\footnotetext{Higher rates may be possible by optimizing power
allocations, for example via water-filling.  Here we only consider
rates achievable using equal-power allocations.}

Recall that $\mat{H}$ at time~$t_0$ is only grouped with the channel
realizations of the alignment set which completes first, so that the
realizations corresponding to the other alignment sets are \emph{not}
associated with $\mat{H}$ and can be used for some other
transmissions.  For example, if $\mat{H}^c$ occurs between times
$\hat{t}_1$ and $\hat{t}_2$ (\ie, $t_0 < \hat{t}_1 < t_1 < \hat{t}_2 <
\hat{t}_3$), then since the transmitter knows the sequence of channel
realizations in advance, it may avoid utilizing $\mat{H}(\hat{t}_1)$
to send $\vect{X}(t_0)$, which would become a wasted transmission when
$\mat{H}^c$ occurs at time $t_1$.  In this example, decoding is via
channel complements, so $\vect{X}(t_0)$ is sent during times $t_0$ and
$t_1$, but never during times $\hat{t}_1$, $\hat{t}_2$, and
$\hat{t}_3$.

We now determine the probability that the first-to-complete scheme
decodes using the alignment set of size~$4$ rather than the alignment
set of size~$2$.  This can be computed by considering a Markov chain
with the following states:
\begin{center}
\footnotesize
\begin{tabular}{ll}
$s_{-1}$: & Decode using $\mat{H}$ and its complement, $\mat{H}^c$ \\
$s_0$:    & No matches yet to any alignment set \\
$s_1$:    & First match with size-$4$ alignment set \\
$s_2$:    & Second match with size-$4$ alignment set \\
$s_3$:    & Third match with size-$4$ alignment set, so decode using
$\mathcal{A}_{\mat{H}}$ 
\end{tabular}
\end{center}
The Markov chain is shown in Figure~\ref{fig:markov_chain_2_4}.
States $s_{-1}$ and $s_3$ are absorbing.  Because this is a success
runs Markov chain~\cite{taylor:stochastic_modeling}, its absorption
probabilities and hitting times are known.  The probability of
decoding via the alignment set of size~$4$ is the probability of
absorption at state~$s_3$ starting from state~$s_0$, and is computed
to be $\beta_4 = 1/4$.  Note that $\beta_4$ does not depend on the
number of possible channel realizations, $|\Hcal|$.  This is intuitive
since matrices not belonging to an alignment set do not affect the
probability that one set completes before another.

\begin{figure}[tbp]
\centering
\psfrag{s0}[c][c]{$s_{-1}$}
\psfrag{s1}[c][c]{$s_0$}
\psfrag{s2}[c][c]{$s_1$}
\psfrag{s3}[c][c]{$s_2$}
\psfrag{s4}[c][c]{$s_3$}
\psfrag{p00}[c][c]{\small $1$}
\psfrag{p10}[c][c]{\small $\frac{1}{|\Hcal|}$}
\psfrag{p11}[c][c]{\small $1 - \frac{4}{|\Hcal|} $}
\psfrag{p12}[c][c]{\small $\frac{3}{|\Hcal|} $}
\psfrag{p20}[c][c]{\small $\frac{1}{|\Hcal|}$}
\psfrag{p22}[c][c]{\small $1 - \frac{3}{|\Hcal|}$}
\psfrag{p23}[c][c]{\small $\frac{2}{|\Hcal|}$}
\psfrag{p30}[c][c]{\small $\frac{1}{|\Hcal|}$}
\psfrag{p33}[c][c]{\small $1 - \frac{2}{|\Hcal|}$}
\psfrag{p34}[c][c]{\small $\frac{1}{|\Hcal|}$}
\psfrag{p44}[c][c]{\small $1$}
\includegraphics[width=3.0in]{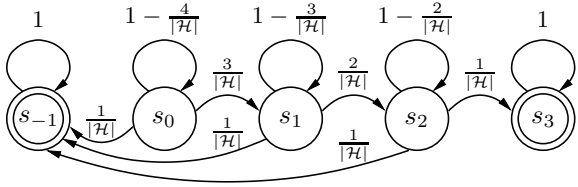}
\caption{Success runs Markov chain associated with first-to-complete
alignment.  States indicate progress towards completion of the
alignment sets.  Quantities above the arrows indicate transition
probabilities.}
\label{fig:markov_chain_2_4}
\end{figure}

\begin{lemma}
\label{thm:rate_first_compl_2_4}
An achievable rate tuple for the first-to-complete scheme has rates
(for all $k = 1,\ldots,K$):
\begin{eqnarray*}
R_k^{(2,4)}
& = & \textstyle \frac{3}{8} E[ \log (1 + 2 |h_{kk}|^2 \snr_k)] \nonumber \\
&   & \textstyle {+} \: \frac{1}{16} E[ \log (1 + 4 |h_{kk}|^2 \snr_k)]
\mbox{.}
\end{eqnarray*}
\end{lemma}
\begin{proof}
Because decoding via the size-$2$ alignment set occurs $1 - \beta_4$
of the time, and decoding via the size-$4$ alignment set occurs
$\beta_4$ of the time, an achievable rate is $R_k^{(2,4)} =
\frac{1}{2} (1-\beta_4) E[ \log (1 + 2 |h_{kk}|^2 \snr_k)] +
\frac{1}{4} \beta_4 E[ \log (1 + 4 |h_{kk}|^2 \snr_k)]$.  Plugging in
$\beta_4 = 1/4$ gives the result.
\end{proof}

\begin{lemma}
\label{thm:delay_first_compl_2_4}
For the first-to-complete scheme, the average decoding delay is
$d^{(2,4)} = (3/4) |\Hcal| = (3/4) d^{(2)}$.
\end{lemma}
\begin{proof}
The delay until either alignment set is completed is the mean hitting
time until one of the corresponding absorption states is reached in
the Markov chain of Figure~\ref{fig:markov_chain_2_4}.  A simple
computation for the hitting time yields $d^{(2,4)} = (3/4) |\Hcal|$.
\end{proof}

\subsection{Delay-Rate Tradeoff}
\label{subsec:align_multiple_tradeoff}

Although the first-to-complete scheme achieves lower delay than
interference alignment using only complements, it has the
drawback of having lower rate.  By using time-sharing, we can achieve
any delay~$d$ such that $(3/4) |\Hcal| = d^{(2,4)} \leq d \leq
d^{(2)} = |\Hcal|$, and every user $k \in \{1,\ldots,K\}$ will still
have increased data rate over that of $R_k^{(2,4)}$.

In the time-sharing scheme, with probability $1 - \alpha$ where $0
\leq \alpha \leq 1$, pair up $\mat{H}$ with the first instance of
$\mat{H}^c$ which occurs later in time; this is alignment using only
complements.  With probability~$\alpha$, however, perform the
first-to-complete scheme: pair up $\mat{H}$ with $\mat{H}^c$ only if
$\mat{H}^c$ occurs before any alignment set of size~$4$ is completed;
otherwise, group $\mat{H}$ with the size-$4$ alignment set which
completes first.

\begin{theorem}
\label{thm:rate_time_sharing}
The achievable rate when time-sharing with probability~$\alpha$ of
using the first-to-complete scheme is
\begin{eqnarray*}
R_k(\alpha)
& = & (1-\alpha) R_k^{(2)} + \alpha R_k^{(2,4)} \\
& = & \textstyle \frac{1}{2} \left( 1 - \frac{\alpha}{4} \right) E[
      \log (1 + 2 |h_{kk}|^2 \snr_k)] \nonumber \\ 
&   & \textstyle {+} \: \frac{\alpha}{16} E[ \log (1 + 4 |h_{kk}|^2 \snr_k)]
\mbox{.}
\end{eqnarray*}
\end{theorem}
\begin{proof}
Evident.
\end{proof}

\begin{theorem}
\label{thm:delay_time_sharing}
The average delay when time-sharing is
\[
d(\alpha) = (1-\alpha) d^{(2)} + \alpha d^{(2,4)} = (1 - \alpha/4)
|\Hcal|
\mbox{.}
\]
\end{theorem}
\begin{proof}
Evident.
\end{proof}

\begin{corollary}
\label{thm:time_sharing_better}
The average delay, when time-sharing between the first-to-complete
scheme (using alignment sets of both sizes $2$ and $4$) and
channel-complement alignment, is lower than the average delay when
using only complements.
\end{corollary}
\begin{proof}
By choosing any $\alpha > 0$, we get delay $d(\alpha)$ strictly less
than $|\Hcal| = d^{(2)}$.
\end{proof}

The reduced delay is an intuitive result since the first-to-complete
scheme allows additional opportunities to align, without disallowing
existing opportunities.

\subsection{Extension to Larger Alignment Sets}
\label{subsec:align_multiple_larger_align_sets}

We now extend our analysis to more general collections of alignment
sets.  Consider a finite tuple of positive even numbers $I =
(m_1,m_2,\ldots,m_{|I|})$, possibly with repetitions.  We generalize
first-to-complete alignment by using non-overlapping alignment sets
with sizes dictated by the entries of $I$.  As soon as all members of
any particular alignment set have been seen, we say that that
alignment set has been completed; we transmit and decode using the
particular alignment set.  As an example, the first-to-complete
alignment scheme given in the first part of this section corresponds
to $I = (2,4)$.  For the case of a general tuple $I$, the process is
identical to the multiple subset coupon collecting problem of Chang
and Ross~\cite{chang:multiple_subset_coupon_collecting}, in which
coupons are repeatedly drawn with replacement until any one of several
preordained subsets of coupons have been collected. 

To compute the achievable rates $(R_1^{I}, R_2^{I}, \ldots, R_K^{I})$
and delay~$d^I$ associated with running first-to-complete alignment
among $I$-sized alignment sets, we construct the associated Markov
chain.  The state vector $\vect{s} = (s_1, s_2, \ldots, s_{|I|})$ is
defined so that element $s_i$ counts how many members of the $i$-th
alignment set have already occurred, excluding the initial
matrix~$\mat{H}$.  Initially, the Markov chain is at state $\vect{s} =
\vect{0}$, since no alignment set member aside from $\mat{H}$ has yet
been realized.  At each time~$t$, if $\mat{H}(t)$ is a member of the
${\hat{\imath}}$-th alignment set and has not yet been realized, then
increment $s_{\hat{\imath}} := s_{\hat{\imath}} + 1$.  When
$s_{\hat{\imath}} = m_{\hat{\imath}} - 1$ for some $\hat{\imath}$,
this means that the ${\hat{\imath}}$-th alignment set (of
size~$m_{\hat{\imath}}$) has been completed.  The Markov chain enters
an absorbing state, and the receiver decodes.  Let $V$ denote the set
of absorbing states.  The state transition probabilities are
\[
\small
P_{\vect{s}, \vect{s}'} =
\left\{ \begin{array}{ll}
\frac{m_{\hat{\imath}} - 1 - s_{\hat{\imath}}}{|\Hcal|}
  & s'_{\hat{\imath}} = s_{\hat{\imath}} + 1 ~\mbox{for some}
    ~\hat{\imath},\ldots \\
  & \quad s'_i = s_i ~\mbox{for all} ~i \ne \hat{\imath}, ~\vect{s}
    \not\in V \\
1 - \sum_i \frac{m_i - 1 - s_i}{|\Hcal|}
  & \vect{s}' = \vect{s}, ~\vect{s} \not\in V \\
1 & \vect{s}' = \vect{s}, ~\vect{s} \in V ~\mbox{(absorption)}\\
0 & \mbox{otherwise}
\end{array} \right.
\mbox{.}
\]
Let $\beta_{m}^{I}$ be the probability that the first completed
alignment set is the alignment set of size $m \in I$.  Equivalently,
$\beta_m^{I}$ is the probability that the Markov chain reaches the
absorption state corresponding to the completion of a specific
size-$m$ alignment set.  These absorption probabilities can be
computed via matrix inversion (see the Appendix or Taylor and
Karlin~\cite{taylor:stochastic_modeling} for more details).
Table~\ref{table:sample_absorp_prob_delays} gives example values for
$\beta_m^{I}$.

\begin{table}[tbp]
\centering
\begin{threeparttable}
\caption{Absorption Probabilities and Delays\tnote{$\dagger$}}
\label{table:sample_absorp_prob_delays}
\begin{tabular*}{0.9\columnwidth}{@{\extracolsep{\fill}} lllll @{}}
\toprule
\multicolumn{1}{c}{Set sizes} & \multicolumn{3}{c}{Absorption
probability} & \multicolumn{1}{c}{Delay} \\
\cmidrule{1-1} \cmidrule{2-4} \cmidrule{5-5}
\multicolumn{1}{c}{$I$} & \multicolumn{1}{c}{$\beta_{m_1}^I$} &
\multicolumn{1}{c}{$\beta_{m_2}^I$} &
\multicolumn{1}{c}{$\beta_{m_3}^I$} & \multicolumn{1}{c}{$d^I$} \\
\midrule
$(2,4)$ & $0.75$ & $0.25$ & & $0.75 |\Hcal|$ \\
$(2,6)$ & $0.8333$ & $0.1667$ & & $0.8333 |\Hcal|$ \\
$(2,4,4)$ & $0.6429$ & $0.1786$ & $0.1786$ & $0.6429 |\Hcal|$ \\
$(2,4,6)$ & $0.6944$ & $0.2083$ & $0.0972$ & $0.6944 |\Hcal|$ \\
$(4,4)$ & $0.5$ & $0.5$ & & $1.2167 |\Hcal|$ \\
$(4,6)$ & $0.625$ & $0.375$ & & $1.3988 |\Hcal|$ \\
$(4,8)$ & $0.7$ & $0.3$ & & $1.4972 |\Hcal|$ \\
$(4,4,4)$ & $0.3333$ & $0.3333$ & $0.3333$ & $0.9790 |\Hcal|$ \\
$(6,10)$ & $0.6429$ & $0.3571$ & & $1.8607 |\Hcal|$ \\
\bottomrule
\end{tabular*}
\begin{tablenotes}
\item[$\dagger$] For values to be valid, $|\Hcal| \geq 1 +
\sum_{i=1}^{|I|} (m_i - 1)$ must hold.
\end{tablenotes}
\end{threeparttable}
\end{table}

Following a similar argument as in
Lemma~\ref{thm:rate_first_compl_2_4}, the rate for receiver $k \in
\{1,\ldots,K\}$ by using a first-to-complete scheme with specific
alignment sets of sizes drawn from $I$ is
\[
R_k^{I} = \sum_{m \in I} \frac{1}{m} \beta_m^{I} E[ \log (1 + m
|h_{kk}|^2 \snr_k)]
\mbox{.}
\]
We now incorporate time-sharing and describe the delay-rate tradeoff.
Let $\mathcal{I}$ be a finite collection of these tuples~$I$; that is,
$\mathcal{I} \subseteq \{I = (m_1,\ldots,m_{|I|}) \; | \; m_i \in 2
\integers^+\}$.  We can do time-sharing between first-to-complete
schemes, with sizes drawn from $I \in \mathcal{I}$, according to the
vector $\vect{\alpha} = (\alpha_{I_1}, \alpha_{I_2}, \ldots,
\alpha_{I_{|\mathcal{I}|}})$ where $\sum_{I \in \mathcal{I}} \alpha_I
= 1$ and $\alpha_I \geq 0$ for all $I \in \mathcal{I}$.  The rate will
be
\begin{equation}
R_k(\vect{\alpha})  = \sum_{I \in \mathcal{I}} \alpha_I R_k^{I}
\mbox{.}
\label{eq:rate_larger_align_sets_time_sharing}
\end{equation}
Alternatively, to be explicit about the rates due to alignment sets of
particular sizes, the rate can also be written as
\[
\small
R_k(\vect{\alpha}) = \sum_{m \in 2 \integers^+} \left( \sum_{I \in
\mathcal{I} \, : \, m \in I} \alpha_I \beta_m^{I} \right) \frac{1}{m}
E[ \log(1 + m |h_{kk}|^2 \snr_k)]
\mbox{.}
\]
The average delay using alignment sets of sizes $I = (m_1, m_2,
\ldots, m_{|I|})$ is equal to the mean absorption time for the Markov
chain.  From \cite{chang:multiple_subset_coupon_collecting}, by using
Poisson embedding, this delay can be computed as\footnotemark
\begin{equation}
d^I = |\Hcal| \int_0^1 \frac{1}{1-u} \prod_{i=1}^{|I|} (1 - u^{m_i -1})
\; du 
\mbox{.}
\label{eq:delay_larger_align_sets}
\end{equation}
Table~\ref{table:sample_absorp_prob_delays} gives average delays for
some representative collections of alignment sets.  Then the delay
using time-sharing is
\begin{equation}
d(\vect{\alpha}) = \sum_{I \in \mathcal{I}} \alpha_I d^{I}
\mbox{,}
\label{eq:delay_larger_align_sets_time_sharing}
\end{equation}
which is linear in the number of possible channel realizations,
$|\Hcal|$.
\footnotetext{This evaluates to an inclusion-exclusion sum of harmonic
numbers $H_n$:
\[
d^I = |\Hcal| \left[ \sum_{U \subseteq I, U \neq \emptyset}
(-1)^{1-|U|} H_{\scriptstyle \left( -|U| + \sum_{m \in U} m \right)}
\right]
\mbox{.}
\]
The delay can also be expressed analytically using the digamma
function $\Psi$, giving
\begin{eqnarray*}
d^I
& = & |\Hcal| \sum_{U \subseteq I} (-1)^{1-|U|} \Psi \left(1-|U| +
\sum_{m \in U} m \right) \\
& = & |\Hcal| \left[
\gamma + \sum_{m_1 \in I} \Psi (m_1) - \sum_{m_1,m_2 \in I} \Psi
(-1+m_1+m_2) \right. \nonumber \\ 
&   & \left. {+} \: \sum_{m_1,m_2,m_3 \in I} \Psi (-2+m_1+m_2+m_3) -
\cdots \right]
\mbox{,}
\end{eqnarray*}
where $\gamma$ is the Euler-Mascheroni constant.  Also, from
\cite{chang:multiple_subset_coupon_collecting}, we can find the
variance of this delay, as well as the average delay when alignment
sets overlap.}

From Table~\ref{table:sample_absorp_prob_delays}, we can make an
observation regarding the computed absorption probabilities and
associated delays.  When the first alignment set has size~$2$, notice
that $d^I = \beta_2^I |\Hcal|$.  This holds for any tuple $I$ which
contains an alignment set of size~$2$ (see Appendix).

\subsection{Further Considerations}
\label{subsec:align_multiple_further}

In this analysis, we only consider alignment sets that do not share
any common matrices.  However, as the number of allowable sizes,
$|I|$, grows larger, this condition will become harder to fulfill
since there will be greater potential for collisions.  Finding tuples
of alignment sets such that there are no overlapping channels is an
avenue for future work.  One thing to note is that because only
$2^{K(K-1)}$ matrices satisfy $h_{kk}^{(i)} = h_{kk}$ and
$|h_{kl}^{(i)}| = |h_{kl}|$ for $k = 1,\ldots,K$ and $l \ne k$, an
alignment set of size $m = 2^{K(K-1)}$ would consist of all possible
channel matrices which might align with $\mat{H}$, and so necessarily
must collide with any other alignment set.

A related issue is that of allowing decoding using \emph{all}
alignment sets of a particular size~$m$, of which there are
$\binom{m-1}{m/2}^{K(K-1)}$ such alignment sets.  For example, a
system could choose to perform first-to-complete alignment among
\emph{any} alignment set of sizes $2$ and $4$.  Because
non-intersection between different alignment sets may no longer be
guaranteed, the analysis will be more complicated.

From Table~\ref{table:sample_absorp_prob_delays}, we can start to
notice the potential for delay reduction via using multiple alignment
sets of the same size.  Although the delay will still scale linearly
in $|\Hcal|$, it is possible to significantly reduce the delay below
$d^{(2)} = |\Hcal|$.  As an example, from
Figure~\ref{fig:multiple_4_delay} we can observe the behavior of the
linear scaling factor, in the case of allowing alignment using more
and more size-$4$ alignment sets.\footnotemark ~Thus a deeper
consideration of alignment with multiple same-size alignment sets may
be a fruitful area for further inquiry.
\footnotetext{Of course, the trend shown in the
Figure~\ref{fig:multiple_4_delay} only holds for scenarios where the
number of users~$K$ is large enough that there exists enough distinct
alignment sets of size~$4$ for alignment.}

\begin{figure}[tbp]
\centering
\psfrag{num4}[c][c]{\scriptsize $n$}
\psfrag{delay}[c][c]{\scriptsize $d^I / |\Hcal|$}
\includegraphics[width=2.5in]{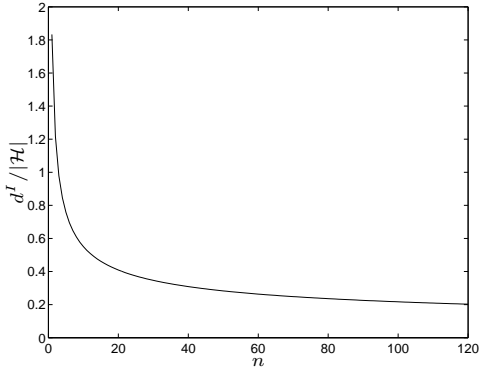}
\caption{Plot showing the decrease of the delay linear scaling factor,
for multiple disjoint alignment sets of size~$4$.  The number of
alignment sets of size~$4$ is $n$.  Thus each point represents the
delay associated with the tuple $I = (4,4,4,\ldots)$, where the tuple
has $n$ elements.}
\label{fig:multiple_4_delay}
\end{figure}

There are myriad other ways in which alignment may occur; \ie, there
is more than one way to align channel matrices.
Definition~\ref{def:align_set} gives one set of sufficient conditions
for channel realizations to align, in order to keep the analysis
tractable---and the benefits which arise by considering larger
alignment sets are already evident.  An obvious extension to this
would be to consider alignment sets in which arbitrary linear
combinations add up to multiples of the identity, and to only consider
alignment among subsets of users.  Subsequent work by
\cite{johnson:drt} takes a step in this direction.

The moral of this story, however, is that delay can always be reduced
by allowing alignment using a greater number of possible choices of
alignment sets.  The data rate may decrease correspondingly, so the
tradeoff needs to be appropriately chosen according to the needs of
the communication system.

\section{Conclusion}
\label{sec:conclusion}

In our analysis, we have not considered the delays between when a
message symbol is available and when it is first transmitted.  We have
only defined delay as the time between when the symbol is first
transmitted and when it is able to be recovered by the receiver.  We
believe this is a reasonable metric of delay, as long as message
symbols are not all generated at one time.  However, an analysis using
queueing theory may be necessary to verify this claim.

In this work, we have proposed an interference alignment scheme which
reduces delay, although with potentially decreased data rate.  Delay
is mitigated by allowing more ways to align interference---through the
utilization of larger alignment sets.  We have also introduced a
scheme to trade off the delay and rate.  In the end, even though the
rate may be reduced, we can still say, in the parlance of interference
aligners, that each person gets $\kappa$ of the cake, where $1/K \leq
\kappa \leq 1/2$---so our scheme can still be an improvement over
non-aligning channel-sharing strategies in terms of data rate.

\addtolength{\textheight}{-7cm}  


\appendix[Markov Chain Analysis]
\label{appendix:markov_chain_analysis}

We provide more details on computing the absorption probabilities and
hitting times from the Markov chain constructions of
Section~\ref{sec:align_multiple}, using techniques from
\cite{taylor:stochastic_modeling}.  Assume  there are a total
of $n$ states in the Markov chain, with $k$ transient states and $n-k$
absorbing states.  In the rest of the appendix, let $\vect{e}_i$
denote a vector consisting of all $0$'s except for a $1$ in the $i$-th
position (\ie, $\vect{e}_i$ is the canonical basis vector in the
$i$-th direction).  We let state $i=0$ be the initial state of the
Markov chain---with no alignment sets completed---so $\vect{e}_0$ is
the initial probability distribution.  Also, let $\ones$ be the
all-ones vector (of appropriate length).

Consider the $n \times n$ probability transition matrix $\mat{P}$,
with the $P_{ij}$ entry denoting the probability of transitioning from
state $i$ to state $j$.  Without loss of generality, we may re-order
the states so that the transient states are indexed first, and then
followed by the absorbing states.  Equivalently, we permute the rows
and columns of $\mat{P}$ to have the block upper-triangular form
$\mat{P} = \left[ \begin{array}{cc}  \mat{Q} & \mat{R} \\ \mat{0} &
\mat{I} \end{array}\right]$, where the block $\mat{Q}$ (of size $k
\times k$) corresponds to transition probabilities between transient
states and the block $\mat{R}$ (of size $k \times (n-k)$) corresponds
to transition probabilities from transient states to absorbing states.
(The lower-right block is the identity matrix since an absorbing state
can only transition to itself, and obviously the lower-left block is
all zeros since absorbing states can not transition to transient
states.)  As an example, if we consider the Markov chain of
Figure~\ref{fig:markov_chain_2_4} with re-ordered state vector
$\vect{s} = (s_0, s_1, s_2, s_3, s_{-1})$, then the permuted
probability transition matrix is
\[
\mat{P} = \left[ \begin{array}{ccccc}
1-\frac{4}{|\Hcal|} & \frac{3}{|\Hcal|} & 0 & 0 & \frac{1}{|\Hcal|} \\
0 & 1-\frac{3}{|\Hcal|} & \frac{2}{|\Hcal|} & 0 & \frac{1}{|\Hcal|} \\
0 & 0 & 1-\frac{2}{|\Hcal|} & \frac{1}{|\Hcal|} & \frac{1}{|\Hcal|} \\
0 & 0 & 0 & 1 & 0 \\
0 & 0 & 0 & 0 & 1
\end{array} \right]
\mbox{,}
\]
which evidently has the appropriate structure.

Expressions for the absorption probabilities and hitting times can be
derived using the various blocks of the probability transition matrix.
\begin{lemma}
\label{thm:markov_chain_absorption_probs}
Define the length-$(n-k)$ absorption probability vector
$\vect{\beta}$, where the $\beta_j$ entry is the probability of
becoming absorbed in state $j$.  Then
\[
\vect{\beta} = (\vect{e}_0^T (\mat{I} - \mat{Q})^{-1} \mat{R})^T
\mbox{.}
\]
\end{lemma}
\begin{proof}
We consider the $k \times (n-k)$ transient-to-absorbing matrix
$\mat{U}$, where the $U_{ij}$ entry denotes the probability of
starting in transient state $i$ and ultimately becoming absorbed in
absorbing state $j$.  By first-step analysis, $\mat{U}$ satisfies the
recursion $\mat{U} = \mat{Q} \mat{U} + \mat{R}$, so $\mat{U} =
(\mat{I}-\mat{Q})^{-1} \mat{R}$.  Because $\mat{Q}$ represents the
probabilities of transitioning between transient states, $(\mat{I} -
\mat{Q})^{-1}$ is the fundamental matrix and is well-defined.  Then
$\beta_j$ is the probability of starting in state $0$ and eventually
becoming absorbed in state $j$, so $\vect{\beta} = (\vect{e}_0^T
\mat{U})^T = (\vect{e}_0^T (\mat{I} - \mat{Q})^{-1} \mat{R})^T$.
\end{proof}

\begin{lemma}
\label{thm:markov_chain_hitting_time}
The hitting time (\ie, the time until absorption in \emph{any}
absorption state) is given by
\[
d = \vect{e}_0^T (\mat{I} - \mat{Q})^{-1} \ones
\mbox{.}
\]
\end{lemma}
\begin{proof}
Let $\vect{D}$ be the length-$k$ vector where the $D_i$ entry is the
hitting time when starting in transient state $i$.  Then first-step
analysis gives the recursion $\vect{D} = \mat{Q} \vect{D} + \ones$, so
$\vect{D} = (\mat{I} - \mat{Q})^{-1} \ones$.  The overall hitting time
is then $d = \vect{e}_0^T \vect{D} = \vect{e}_0^T (\mat{I} -
\mat{Q})^{-1} \ones$.
\end{proof}

Suppose one employs the first-to-complete alignment scheme with
alignment sets of sizes $I = (m_1,m_2,\ldots,m_{|I|})$, and where the
first alignment set has size $m_1 = 2$.  Here we prove that the mean
time to absorption of the Markov chain is equal to the number of
possible channel fading matrices multiplied by the probability of
completion using the set of size~$2$.
\begin{theorem}
\label{thm:size_2_hitting_time}
If $2 \in I$, then $d^I = \beta_2^I |\Hcal|$.
\end{theorem}
\begin{proof}
Let $\hat{\jmath}$ be the state associated with the realization of the
channel complement (\ie, the state associated with completing the
size-$2$ alignment set).  We assume that each channel realization is
equally likely with probability $1/|\Hcal|$, so the probability of
transitioning into state $\hat{\jmath}$ is $1/|\Hcal|$ starting from
any [transient] state.  From
Lemma~\ref{thm:markov_chain_absorption_probs} and since the
$\hat{\jmath}$-th column of $\mat{R}$ is $(1/|\Hcal|) \ones$, we see
that $\beta_2^I = \beta_{\hat{\jmath}} = (1/|\Hcal|) \vect{e}_0^T
(\mat{I} - \mat{Q})^{-1} \ones$.  Since $d^I = \vect{e}_0^T (\mat{I} -
\mat{Q})^{-1} \ones$, the result follows.
\end{proof}

\section*{Acknowledgements}
\label{sec-ack}

This material is based upon work performed while J.~Koo was supported
by AFOSR grant FA9550-06-1-0312 and W.~Wu was supported by the
Frank~H.~Buck Scholarship.  The authors would like to thank Michael
Gastpar for helpful comments regarding the ergodic alignment scheme
and for pointers to related work.  The authors would also like to
acknowledge anonymous reviewers, whose comments helped to improve the
clarity of the exposition.

\bibliographystyle{IEEEtran}
\bibliography{IEEEabrv,references}

\end{document}